\documentclass[11pt]{article}
\usepackage{a4wide}
\usepackage{amsthm, amssymb}
\usepackage{amsmath}
\newtheorem{lemma}{Lemma}
\newtheorem{theorem}{Theorem}
\newtheorem{corollary}{Corollary}

\begin{document}
\title{Faster Separators for Shallow Minor-Free Graphs via Dynamic Approximate Distance Oracles}
\author{Christian Wulff-Nilsen~\footnote{Department of Computer Science, University of Copenhagen, koolooz@di.ku.dk, ,
                  \texttt{http://www.diku.dk/$_{\widetilde{~}}$koolooz/}.}}
\maketitle


\begin{abstract}
Plotkin, Rao, and Smith (SODA'$97$) showed that any graph with $m$ edges and $n$ vertices that excludes $K_h$ as a depth $O(\ell\log n)$-minor
has a separator of size $O(n/\ell + \ell h^2\log n)$ and that such a separator can be found in $O(mn/\ell)$ time. A time bound
of $O(m + n^{2+\epsilon}/\ell)$ for any constant $\epsilon > 0$ was later given (W., FOCS'$11$) which is an improvement for non-sparse graphs.
We give three new algorithms. The first has the same separator size and running time $O(\mbox{poly}(h)\ell m^{1+\epsilon})$. This is a significant
improvement for small $h$ and $\ell$. If $\ell = \Omega(n^{\epsilon'})$ for an arbitrarily small chosen constant $\epsilon' > 0$, we get a
time bound of $O(\mbox{poly}(h)\ell n^{1+\epsilon})$.
The second algorithm achieves the same separator size (with a slightly larger polynomial dependency on $h$) and running time
$O(\mbox{poly}(h)(\sqrt\ell n^{1+\epsilon} + n^{2+\epsilon}/\ell^{3/2}))$ when $\ell = \Omega(n^{\epsilon'})$.
Our third algorithm has running time $O(\mbox{poly}(h)\sqrt\ell n^{1+\epsilon})$ when $\ell = \Omega(n^{\epsilon'})$. It finds a separator of size
$O(n/\ell) + \tilde O(\mbox{poly}(h)\ell\sqrt n)$\footnote{We use $\tilde O$-notation instead of $O$-notation when
suppressing polylogarithmic factors.} which is no worse than previous bounds when $h$ is fixed and
$\ell = \tilde O(n^{1/4})$. A main tool in obtaining
our results is a novel application of a decremental approximate distance oracle of Roditty and Zwick.
\end{abstract}

\section{Introduction}
Given an undirected graph with a non-negative vertex weight function, a separator of $G$ is a subset of vertices whose removal
partitions $G$ into connected components none of which contain more than a $c$-fraction of the total vertex weight of the graph,
where $c < 1$ is a constant. Clearly, any graph contains a separator (the entire vertex set). The goal is to find separators of
small size.

A celebrated theorem of Lipton and Tarjan~\cite{SeparatorPlanar} states that every planar graph of size $n$ has an
$O(\sqrt n)$-size separator which can be found in linear time. Alon, Seymour, and Thomas~\cite{SeparatorHMinor} generalized this
to minor-free graphs.

In this paper we study the bigger class of shallow minor-free graphs. A depth $\ell$-minor of a graph is a minor where every vertex
corresponds to a contracted subgraph of radius at most $\ell$. This class was introduced by
Plotkin, Rao, and Smith~\cite{ShallowMinor}. They gave applications of these graphs to, e.g., geometry. Shallow minors are used to distinguish between somewhere and nowhere dense graphs~\cite{Sparsity}.

Plotkin, Rao, and Smith showed that an $n$-vertex graph excluding $K_h$ as a depth $O(\ell\log n)$-minor
has a separator of size $O(n/\ell + \ell h^2\log n)$. They gave an algorithm with $O(mn/\ell)$ running time which
either outputs a depth $O(\ell\log n)$-minor or a separator of size $O(n/\ell + h^2\ell\log n)$. For non-sparse graphs,
running time was improved to $O(n^{2+\epsilon}/\ell)$ for an arbitrarily small constant $\epsilon > 0$ with only a constant-factor
increase (depending on $\epsilon$) in the separator size~\cite{CWN}.

We give three new algorithms to find separators in shallow minor-free graphs. The first achieves the same separator size as
in~\cite{ShallowMinor} but with a running time of $O(h\ell m^{1+\epsilon} + h^2n\log n)$ for an arbitrarily small constant
$\epsilon > 0$. For small $h$ and $\ell$, this is near-linear time and a significant improvement over the near-quadratic time
in~\cite{CWN}. If $\ell = \Omega(n^{\epsilon'})$ for an arbitrarily small constant $\epsilon' > 0$, $m$ can be replaced by $n$ and $h$ can be replaced by $\mbox{poly}(h)$ in the time bound due to sparsity of the graph.
Our second algorithm achieves essentially the same separator size in time
$O(\mbox{poly}(h)(\sqrt\ell n^{1+\epsilon} + n^{2+\epsilon}/\ell^{3/2}))$ when $\ell = \Omega(n^{\epsilon'})$. The special case $\ell = \sqrt n$ gives the same time bound for
minor-free graphs as in~\cite{CWN}. Our third algorithm is the fastest, achieving a running time of
$O(\mbox{poly}(h)\sqrt\ell n^{1+\epsilon})$ when $\ell = \Omega(n^{\epsilon'})$. It finds a separator of size $O(n/\ell) + \tilde O(\mbox{poly}(h)\ell\sqrt n)$ which asymptotically is no
worse than the size achieved in~\cite{ShallowMinor} when $\ell = \tilde O(n^{1/4})$ and $h = O(1)$.
A main tool to achieve these new time bounds is a novel application of a dynamic approximate distance oracle of Roditty and
Zwick~\cite{DecDistOracle}. We believe this connection is interesting in itself and should further motivate the study of
dynamic distance oracles, an area which has only recently received attention from the research community.

The organization of the paper is as follows. We first give basic definitions, notation, and results in
Section~\ref{sec:Prelim}. Then we give a generic algorithm in Section~\ref{sec:Generic} which is quite similar to that of
Plotkin, Rao, and Smith. All our algorithms are implementations of this generic algorithms. Then our three algorithms are given
in Sections~\ref{sec:SepDistOracle},~\ref{sec:Spanner}, and~\ref{sec:SpannerDistOracle}, respectively. Finally, we conclude in
Section~\ref{sec:ConclRemarks}. Some details have been moved to the appendix.

\section{Preliminaries}\label{sec:Prelim}
We consider undirected graphs only. For a graph $G$, $V(G)$ resp.~$E(G)$ denotes the vertex set resp.~edge set of $G$ and for a subset $X$ of $V$, $G[X]$ is the subgraph of $G$ induced by $X$.
Connected components are referred to simply as \emph{components}. For vertices $u,v\in V$,
$d_G(u,v)$ denotes the shortest path distance between $u$ and $v$ in $G$. This definition is extended to
edge-weighted graps.

Consider a graph $G = (V,E,w:V\rightarrow\mathbb R)$ with a non-negative vertex weight function $w$. For any
subset $X$ of $V$, let $w(X) = \sum_{v\in X}w(v)$.
A \emph{separator} of $G$ is a subset $S$ of $V$ such that for each component $C$ of $G[V\setminus S]$
we have $w(C)\leq cw(V)$, for some constant $c < 1$.

Given graphs $G$ and $H$, $H$ is a \emph{minor} of $G$ if $H$ can be obtained from a subgraph of $G$ by
edge contractions. Otherwise, we say that
$G$ is \emph{$H$-minor-free} or that $G$ \emph{excludes $H$ as a minor}. If $H$ is a minor of $G$
such that the subgraphs of $G$ corresponding to vertices of $H$ have radius
at most $\ell$, $H$ is a \emph{minor of depth $\ell$}. If not, $G$ \emph{excludes $H$ as a minor of depth $\ell$}.
We only consider complete excluded minors, i.e., $H$ is of the form $K_h$. For our application, this is without loss
of generality since if $G$ excludes a minor $H$ (of some depth), it also excludes $K_h$ as a minor (of some depth), where $h = |V(H)|$; furthermore, our bounds only have small polynomial dependency on the size of the excluded minor.
For integers $h\geq 1$ and $\ell\geq 0$, let $\mathcal G_{h,\ell}$ denote the set of graphs that exclude $K_h$ as a minor of depth $\ell$. We have the following
simple lemma
\begin{lemma}\label{Lem:ShallowInclusions}
For any $h$, $\mathcal G_{h,0}\supseteq\mathcal G_{h,1}\supseteq\mathcal G_{h,2}\supseteq\ldots$.
\end{lemma}
An essential tool in our algorithms is a decremental approximate distance oracle of Roditty and
Zwick which can report approximate distances for close vertex pairs~\cite{DecDistOracle}. The following lemma states the
performance of this oracle.
\begin{lemma}\label{Lem:DistOracle}
Let $k,d\in\mathbb N$ and let $G$ be a given graph with integer edge weights and with $m$ edges and $n$
vertices. There is a data structure
of size $O(m + n^{1 + 1/k})$ which can be maintained under edge deletions in $G$ in $O(dmn^{1/k})$ total expected time such that after each edge
deletion, for any vertices $u,v$ in $G$, an estimate $\tilde d_G(u,v)$ of the shortest path distance $d_G(u,v)$ can be reported in $O(k)$ time satisfying the
following: If $d_G(u,v)\leq d$ then $d_G(u,v)\leq\tilde d_G(u,v)\leq (2k-1)d_G(u,v)$ and if $d_G(u,v) > d$ then
$d_G(u,v)\leq\tilde d_G(u,v)$. A $uv$-path achieving this distance estimate can be reported in constant time per vertex, starting from either $u$
or $v$.
\end{lemma}
The last part of the lemma is not stated in~\cite{DecDistOracle} but it follows immediately from the observation that the oracle
maintains cluster structures (see definition
in~\cite{ThorupZwick}) and paths can be traversed in constant time per vertex using these.

\section{A Generic Algorithm}\label{sec:Generic}
All of our algorithms are implementations of the same generic algorithm (except the algorithm 
in Section~\ref{sec:Spanner} which implements a slight variant) and we refer to it as
$\texttt{genericalg}(G = (V,E,w:V\rightarrow\mathbb R), h, \ell)$; see pseudocode in Figure~\ref{fig:genericalg}. It is similar to that of Plotkin, Rao,
and Smith~\cite{ShallowMinor} but with some subtle differences that we
come back to later. We refer to~\cite{ShallowMinor} for a less compact description.
\begin{figure}[!ht]
\begin{tabbing}
d\=dd\=\quad\=\quad\=\quad\=\quad\=\quad\=\quad\=\quad\=\quad\=\quad\=\quad\=\quad\=\kill
\>\texttt{Algorithm} $\texttt{genericalg}(G = (V,E,w:V\rightarrow\mathbb R), h, \ell)$\\\\
\>1. \>\>initialize $V' = V$, $M = \mathcal M = V_r = B = A = \emptyset$, and $\rho = 2\lceil\ell\ln n\rceil$\\
\>2. \>\>while there is a component $X$ of $G[V']$ with $w(X) > \frac 2 3 w(V)$ and $|\mathcal M| < h$\\
\>3.\>\>\>update $V'$ to $X$ and move any remaining components to $V_r$\\
\>4. \>\>\>for each $T\in\mathcal M$ that is not incident to $V'$ in $G$\\
\>5. \>\>\>\>$\mathcal M\leftarrow\mathcal M\setminus\{T\}$, $M\leftarrow M\setminus V(T)$, $V_r\leftarrow V_r\cup V(T)$, $A\leftarrow A\cup V(T)$\\
\>6. \>\>\>let $u$ be any vertex in $V'$\\
\>7. \>\>\>if a tree $T$ is identified in $G[V\setminus (M\cup A)]$ rooted at $u$, with radius $\leq C\rho$,\\
\>   \>\>\>(for some constant $C\geq 1$) and incident in $G$ to all trees in $\mathcal M$\\
\>8. \>\>\>\>update $\mathcal M\leftarrow\mathcal M\cup\{T\}$, $M\leftarrow M\cup V(T)$, and $V'\leftarrow V'\setminus V(T)$\\
\>9. \>\>\>else // it is assumed that $v$ can be found in line $10$\\
\>10.\>\>\>\>pick a $v\in V'$ which is incident in $G$ to $M$ such that $d_{G[V']}(u,v) > \rho$\\
\>11.\>\>\>\>letting $w$ be any of $u$ and $v$, start growing a BFS tree in $G[V']$ from $w$\\
\>12.\>\>\>\>for each BFS layer $N$ explored\\
\>13.\>\>\>\>\>let $S$ be the set of vertices explored in all layers so far\\
\>14.\>\>\>\>\>let $S'$ be the set of smaller vertex weight among $S$ and $V'\setminus (S\setminus N)$\\
\>15.\>\>\>\>\>if $|N|\leq |S'|/\ell$\\
\>16 \>\>\>\>\>\>update $B\leftarrow B\cup N$, $V_r\leftarrow V_r\cup S'\setminus N$, and $V'\leftarrow V'\setminus S'$\\
\>17 \>\>\>\>\>\>terminate the BFS\\
\>18.\>\>if $|\mathcal M| = h$ then output $\mathcal M$ as a $K_h$-minor of $G$ of depth $O(\ell\log n)$\\
\>19.\>\>else output $M\cup B$ as a separator of $G$\\
\end{tabbing}
\caption{A generic algorithm for either finding a separator or reporting $K_h$ as a depth $O(\ell\log n)$-minor in a
         vertex-weighted graph $G = (V,E,w)$ where $n = |V|$.}\label{fig:genericalg}
\end{figure}
Subsets $V'$, $M$, $V_r$, $B$, and $A$ of $V$ are maintained where initially $V' = V$ and
$M = V_r = B = A = \emptyset$~\footnote{Here
we use a similar naming convention as in~\cite{ShallowMinor}.}. The algorithm attempts
to form a separator of $G$ in a number of iterations of the while-loop in lines $2$--$17$. Set $M$ is the union of vertex sets of trees and
$\mathcal M$ denotes the
set of these trees. At any given time, the size $p$ of $\mathcal M$ is at most $h$ and the trees of $\mathcal M$ form the minor
$K_p$ of depth $O(\ell\log n)$ in $G$. At the beginning of each iteration, trees that are no longer incident to $V'$ in $G$ are removed
from $\mathcal M$ (lines $4$--$5$). Set $A$ is monotonically growing and consists at any given time of the vertices that
previously belonged to $M$ but have since been removed from this set. Set $V'$ is a monotonically shrinking set and consists of vertices
yet to be processed. The algorithm terminates if each component of $G[V']$ has vertex weight at most $\frac 2 3 w(V)$.

In the following, let $\rho = 2\lceil\ell\ln n\rceil$. In each iteration, the algorithm attempts to find a tree $T$ in $G[V\setminus A]$ which is rooted at an
arbitrary $u\in V'$, has radius at most $C\rho$ for a constant $C\geq 1$, and is incident in $G$ to all trees in
$\mathcal M$. If such a tree is found, the trees in $\mathcal M\cup\{T\}$ form the minor $K_{p+1}$ of depth $O(\ell\log n)$ in $G$. In this case,
$T$ is added to $\mathcal M$. If $p+1 = h$, the algorithm outputs the $p+1$ trees found as they constitute a certificate that $G$ contains $K_h$ as
a minor of depth $O(\ell\log n)$.

Suppose such a tree $T$ could not be found. Then it is assumed in line $10$ that a vertex $v\in V'$ incident in $G$ to $M$ exists such that $d_{G[V']}(u,v) > \rho$.
(this assumption holds for the algorithm in~\cite{ShallowMinor} with $C = 1$ since in this case there is a tree $T'\in\mathcal M$ such that for any $v\in V'$ incident to $T'$ in $G$,
$d_{G[V\setminus A]}(u,v) > C\rho$.) Then either $u$ or $v$ is
picked; call it $w$. A BFS tree is grown from $w$ in $G[V']$ until some layer $N$ is small compared to $S$ and $V'\setminus (S\setminus N)$, where $S$
is the set of vertices explored so far. More precisely, the process stops if $|N|\leq |S'|/\ell$, where $S'$ is the set of smaller vertex weight among $S$ and
$V'\setminus S$. Using the same analysis as in~\cite{ShallowMinor}, such a layer $N$ will eventually be found. For assume otherwise. Since $d_{G[V']}(u,v) > \rho$, more than $\rho$
layers are explored by the BFS.
For each layer explored, either $|S|$ grows by at least a factor of $1 + 1/\ell$ or $|V'\setminus (S\setminus N)|$ is reduced by at least a factor of $1 + 1/\ell$.
Since $(1 + 1/\ell)^{\rho/2}\geq n$, after $\rho$ layers, either $S$ contains at least $n$ vertices or $V'\setminus (S\setminus N)$ contains less than $1$ vertex,
which is a contradiction.



\paragraph{Correctness:}
To prove correctness of the generic algorithm (under the assumption stated in line $9$), note that the output in line $18$ is indeed a $K_h$-minor of depth $O(\ell\log n)$.
We need to show that in line $19$, $M\cup B$ is a separator. It is easy to see that $V'$ is disjoint from $V_r\cup M\cup B$ during the algorithm.
Consider the final iteration of the while-loop. At the beginning of line $7$,
$w(V') > \frac 2 3 w(V)$ so $w(V_r) < w(V)/3$. If line $8$ is executed then all vertices that are removed from
$V'$ are added to the set output in line $19$. If lines $10$--$17$ are executed then let $W' > \frac 2 3 w(V)$ be the vertex
weight of $V'$ just before executing line $10$ and let
$W_r$ be the vertex weight of $V_r$ just after executing line $17$. By definition of $S'$, $W_r\leq (w(V) - W') + W'/2 = w(V) - W'/2 < \frac 2 3 w(V)$.
Since this is the final iteration, we also have that each component of $G[V']$ has vertex weight less than
$\frac 2 3 w(V)$ at the end of the iteration.
Hence, $M\cup B$ is a separator of $G$ in line $19$.



The main difference between the above generic algorithm and that in~\cite{ShallowMinor} is that the former searches for a tree in $G[V\setminus (M\cup A)]$ and
allows $M$ to overlap with $B$ and $V_r$ whereas the latter does not maintain $A$ and instead searches for a tree in $G[V']$, ensuring that
$(V',M,V_r,B)$ is a partition of $V$. This difference will be important for our algorithm in
Section~\ref{sec:SpannerDistOracle}. Unlike~\cite{ShallowMinor}, we also allow constant $C$ in line $7$ in order to facilitate the use of approximate distances.

It is easy to see that if each tree added to $\mathcal M$ in line $5$ of $\texttt{genericalg}(G = (V,E), h, \ell)$ contains at most $s$ vertices and if a separator is output in line $19$, its size is $O(hs + n/\ell)$.
Note that if a tree $T$ exists in line $7$ of $\texttt{genericalg}(G = (V,E), h, \ell)$, we may pick it such that its size is
$O(h\ell\log n)$. This then gives a separator of size $O(h^2\ell\log n + n/\ell)$, matching that in~\cite{ShallowMinor}.
For our algorithm in Section~\ref{sec:SpannerDistOracle}, we will need to allow a bigger size of each tree $T$. This will increase the
size of the separator but only for large values of $\ell$.

\section{A Fast Separator Theorem via Distance Oracles}\label{sec:SepDistOracle}
For the description of our first algorithm, it will prove useful to regard $\mathcal M$ as consisting of exactly $h-1$ trees at any
given time; this is done by permitting trees with empty vertex sets.
Each tree of $\mathcal M$ is assigned a unique index between $1$ and $h-1$ and we let $T_i$ denote the tree with index $i$. From the moment a tree is
added to $\mathcal M$ until it leaves $\mathcal M$ or
the algorithm terminates, the index of that tree is fixed. We refer to the trees of $\mathcal M$ with non-empty vertex sets as \emph{proper} trees. Each
of them will have diameter $O(\ell\log n)$ and size $O(\ell h\log n)$. Thus, the size of the separator output in line $19$
is $O(h^2\ell \log n + n/\ell)$.

For $1\leq i\leq h-1$, we maintain a decremental approximate distance oracle $\mathcal D_i$ satisfying the conditions in Lemma~\ref{Lem:DistOracle}
with $k = \lceil 1/\epsilon\rceil$ and $d = 4k\rho$. Denote by $V_i'$ the set $(V\setminus (M\cup A))\cup V(T_i)$. At any point, $\mathcal D_i$ is an
approximate distance oracle for the graph with vertex set $V$ and containing all edges of $G[V_i']$ except those with both endpoints in $T_i$ that do
not belong to $T_i$. We denote this graph by $G_i' = (V,E_i')$. Note that the vertex set is fixed to $V$ since $\mathcal D_i$ does not support vertex deletions.

\subsection{Identifying a tree of small diameter}
Now, consider an iteration of the while-loop of $\texttt{genericalg}$. To determine whether to execute line $8$ or lines $10$--$17$, we do as follows.
Assume that at least one tree in
$\mathcal M$ is proper (otherwise, a tree satisfying the conditions in line $7$ can easily be found) and let $i_{\min}$ be the smallest index of
such a tree. Pick $u$ as any vertex of $V'$ incident in $G$ to $T_{i_{\min}}$. For each proper $T_i\in \mathcal M$ with $i > i_{\min}$ and for each $v\in V(T_i)$, we query $\mathcal D_i$ for
a distance estimate from $u$ to $v$ in $G_i'$. Let $v_i\in V(T_i)$ be a vertex $v$ achieving a minimum estimate and denote this estimate by $\tilde\delta_i$.
In the following, we consider the following two possible cases:
\begin{enumerate}
\item $\tilde\delta_i\leq d$ for all $i > i_{\min}$ for which $T_i$ is proper,
\item $\tilde\delta_i > d$ for some $i > i_{\min}$ for which $T_i$ is proper.
\end{enumerate}

We will show how a tree satisfying the conditions in line $7$ can be found if we are in case $1$ and that a vertex $v$ satisfying the conditions in line $10$
can be found if we are in case $2$.

Assume case $1$ first. We traverse the approximate $uv_i$-path from $u$ until reaching
the first vertex which is incident in $G$ to $T_i$. Note that the subpath traversed is contained in $G[V\setminus (M\cup A)]$.
By Lemma~\ref{Lem:DistOracle}, each traversal takes $O(d)$ time and since each approximate
distance query can be done in $O(1)$ time, $v_i$ is found in $O(|V(T_i)|) = O(hd)$ time.
Summing over all $i$, we obtain in $O(h^2d)$ time a tree $T$ in $G[V\setminus(M\cup A)]$ rooted at $u$ of radius at most $d = 4k\rho$ and of size $O(hd)$
which is incident in $G$ to each proper tree of $\mathcal M$ ($T$ is in the union of the paths traversed). This tree satisfies the conditions in line $7$
with $C = 4k$.
We may assume that $T$ contains at least $d$ vertices; if not, grow it by adding additional incident vertices to it from $V\setminus(M\cup A)$ while keeping the radius at most $d$.

Having found $T$, we check if there is an index available for $T$ in $\mathcal M$, i.e., an index $i$ between $1$ and $h-1$ such that $T_i$ is not proper. If there is
no such index, $G$ has
$K_h$ as a minor of depth $O(\ell\log n)$ and the algorithm outputs $T$ and the collection of proper trees in $\mathcal M$ as a certificate of this.
Otherwise, $T_i$ is set to $T$, thereby adding it to $\mathcal M$ and making it proper, and $V(T)$ is removed from $V'$ and added to $M$.
This update may cause some other proper trees $T_j\in \mathcal M$ to no longer be incident to $V'$ in $G$ and each of these
trees need to be moved from $\mathcal M$ to $V_r$ in the next iteration (if any). For this to happen for such a tree $T_j$, each edge of $G_j'$ incident to $T_j$ must either belong to $T_j$
or be incident to the set $V''$ of vertices that will be removed from $V'$ in line $3$ of the next iteration. Note that
since $V'$ is a monotonically decreasing set, we can afford to visit the set $E''$ of edges of $G$ which are incident to $V''$ and mark them. Now for each proper $T_j$,
$j\neq i$, we visit all its vertices and check if it has any incident non-tree edges in $E_j'$ which are not marked; this takes
$O(|T_j| + |E''|) = O(hd + |E''|)$ time. If there is no such vertex, $V(T_j)$ is moved from $M$ to $V_r$ and $T_j$ is updated to be an empty tree in $\mathcal M$.
Over all $j$, this takes $O(h^2d + h|E''|)$ time. The latter term is paid for by initially putting $h$ credits on each edge of $E$.

To update the approximate distance oracles accordingly, consider first moving $V(T) = V(T_i)$ from $V'$ to $M$. This amounts to deleting from $G_j'$ all
edges incident to $V(T_i)$, for each $j\neq i$.
From $\mathcal D_i$, we need to delete edges with both endpoints in $T_i$ that do not belong to $T_i$. We identify these by simply visiting
all edges incident to $T_i$ since this can be paid for by the removal of $V(T_i)$ from $V'$.
For a tree $T_j$ whose vertex set needs to be moved from $M$ to $V_r$, we update $\mathcal D_j$ by deleting all edges incident to $V(T_j)$.

By Lemma~\ref{Lem:DistOracle}, total time to maintain all $h-1$ distance oracles is $O(hdmn^{1/k}) = O(h\ell mn^{\epsilon}\log n)$.
Since each tree removed from $V'$ has size $\Omega(d)$, total additional time for the above is
$O(hm + (n/d)h^2d) = O(hm + h^2n)$.

\subsection{Identifying a small separator for distant vertex pairs}
The following lemma considers case $2$ above.
\begin{lemma}\label{Lem:ShortLongPath}
With the above definitions, suppose $\tilde\delta_i > d$ for some $i > i_{\min}$ where $T_i$ is proper. Then the shortest path distance in $G[V']$ from $u$ to any vertex of $V'$
incident in $G$ to $T_i$ is at least $d/(2k-1) - 1$.
\end{lemma}
\begin{proof}
Let $w_i$ be a vertex of $V'$ incident in $G$ to a vertex $v_i$ in $T_i$ and assume for contradiction that $d_{G[V']}(u,w_i) < d/(2k-1) - 1$.
We have
\[
  d_{G_i'}(u,v_i)\leq d_{G_i'}(u,w_i) + 1\leq d_{G[V']}(u,w_i) + 1 < d/(2k-1)
\]
so $\tilde\delta_i > d > (2k-1)d_{G_i'}(u,v_i)$. By Lemma~\ref{Lem:DistOracle}, $d_{G_i'}(u,v_i) > d$, contradicting that $d_{G_i'}(u,v_i) < d/(2k-1)\leq d$.
\end{proof}
Let $i$ be an index satisfying Lemma~\ref{Lem:ShortLongPath}. The lemma shows that if we are in case $2$ above, the condition in line $10$ holds since
for any vertex $v\in V'$ incident in $G$ to $T_i$, $d_{G[V']}(u,v) > d/(2k-1) - 1 = 4k\rho/(2k-1) - 1 > \rho$. Such a vertex $v$ can be found in $O(|T_i|) = O(hd)$ time.
We now describe how to efficiently execute lines $10$--$17$. Since we are free to perform a BFS from $u$ or from $v$, we
perform both of these searches in parallel. More precisely, one edge is visited by the BFS from $u$, then one edge
by the BFS from $v$, and so on. When the condition in line $15$ is satisfied for either of them, both searches terminate.



To analyze the running time, observe that when the BFS procedures terminate, each has visited at most $\rho - 1 = d/(4k) - 1$ layers. By Lemma~\ref{Lem:ShortLongPath},
$E_u\cap E_v = \emptyset$ where $E_u$ resp.~$E_v$ denotes the set of edges explored by the BFS from $u$ resp.~$v$.
Let $U$ be the vertex set moved from $V'$ to $V_r\cup B$ after the BFS procedures from $u$ and $v$ terminate
and let $t$ denote the BFS time for exploring $E_u$ and $E_v$. Since
$|E_u| = |E_v|\pm 1$, we have that $t = O(|E_u|)$ as well as $t = O(|E_v|)$. Since either
$E_u\subseteq E(G[U])$ or $E_v\subseteq E(G[U])$, we get $t = O(|E(G[U])|)$. Hence, the removal of $U$ from $V'$ can pay for $t$ so total BFS time over all iterations is $O(m)$.

As for case $1$ above, we move a tree $T_j$ from $M$ to $V_r$ if $T_j$ is no longer incident in $G$ to $V'$ in line $3$ of
the next iteration. Letting $V''$ and $E''$ be defined as above, we mark all
edges of $G$ incident to $V''$ and then for each proper $T_j$ check in $O(|T_j|)$ time whether $G_j'$ has any non-marked edges incident to $T_j$ that do not belong to $T_j$. This takes $O(|E''| + h^2d)$ time over all $j$. Since $V''$ has size $\Omega(\ell)$, this is $O(m + (n/\ell)h^2d) = O(m + n\log nh^2)$
over all iterations.

We can now conclude this section with our first main result.
\begin{theorem}\label{Thm:Algo1}
Given a graph $G$ with $m$ edges and $n$ vertices and given $h,\ell\in\mathbb N$ and $\epsilon > 0$, there is an algorithm with
$O(h\ell m^{1+\epsilon} + h^2n\log n)$ running time which either
gives a certificate of $K_h$ as a depth $O(\ell\log n)$-minor in $G$ or outputs a separator of $G$ of size $O(h^2\ell\log n + n/\ell)$.
\end{theorem}
In the appendix, we show that when $\ell = \Omega(n^{\epsilon'})$ for an arbitrarily small chosen constant $\epsilon' > 0$
then $m = O(\mbox{poly}(h)n)$. In this
case, our algorithm can be modified to initially check if $G$ exceeds this edge bound and if so reject and halt. Hence, we can replace $m$ by $n$ in Theorem~\ref{Thm:Algo1}.

\section{Speed-up using Bootstrapping and Spanners}\label{sec:Spanner}
In this section, we speed up the result in Theorem~\ref{Thm:Algo1} for large $\ell$ using a boot-strapping technique
presented in~\cite{CWN}. Note that the running time in Theorem~\ref{Thm:Algo1} grows with $\ell$. The idea for the speed-up
is to apply the theorem for a smaller depth than $\ell$ (which is possible by Lemma~\ref{Lem:ShallowInclusions})
to obtain large separators fast and to use these separators to build
a compact representation of certain subgraphs of $G$
which can be used to give a fast implementation of $\texttt{genericalg}$ for the actual depth $\ell$.
To simplify our bounds, we assume here and in the next section that $h = O(1)$ and $\ell = \Omega(n^{\epsilon'})$ for an
arbitrarily small constant $\epsilon' > 0$. The actual dependency on $h$ is a low-degree polynomial and for smaller values of $\ell$,
we can apply our first algorithm.
We have moved the details for our second algorithm
to the appendix since the technique is basically the same as in~\cite{CWN}.
The performance of our second algorithm is stated in the following theorem.
\begin{theorem}\label{Thm:Algo2}
Given a graph $G$ with $n$ vertices and given $h,\ell\in\mathbb N$ and constants $\epsilon,\epsilon' > 0$ such that
$h = O(1)$ and $\ell = \Omega(n^{\epsilon'})$, there is an algorithm with
$O(n^{1+\epsilon}\sqrt\ell + n^{2+\epsilon}/\ell^{3/2})$ running time which either
correctly reports that $G$ contains $K_h$ as a depth $O(\ell\log n)$-minor or outputs a separator of $G$ of size $O(\ell\log n + n/\ell)$.
\end{theorem}

\section{Combining Spanners and Distance Oracles}\label{sec:SpannerDistOracle}
Our third and final algorithm combines techniques from the other two. The overall idea is to maintain decremental
distance oracles as in the first algorithm but for a graph of sublinear size consisting of spanners as in the
second algorithm. However, there are several obstacles that we need to overcome to handle this, most notably
that maintaining a sublinear-size graph representation require edge insertions in addition to deletions, something that
the oracles do not support.

First, we give some definitions and results that we needed for the second algorithm (see the appendix for more details).
For a connected graph $G = (V,E)$, a \emph{cluster} (of $G$) is a connected subgraph of $G$. A \emph{clustering} (of $G$) is a
collection $\mathcal C$ of clusters of $G$ whose edge sets form a partition of $E$.
A \emph{boundary vertex} of a cluster $C\in\mathcal C$ is
a vertex that $C$ shares with other clusters in $\mathcal C$. All other vertices of $C$ are
\emph{interior vertices} of $C$. For a subgraph $C'$ of $C$, we let $\delta C'$ denote the set of boundary vertices of $C$
contained in $C'$ and we refer to them as the boundary vertices of $C'$ (w.r.t.\ $C$).

Let $n$ be the number of vertices of $G$. For a parameter $r > 0$, an \emph{$r$-clustering} (of $G$) is
a clustering with clusters having a total of $\tilde O(n/\sqrt r)$ boundary vertices (counted with multiplicity) and each
containing at most $r$ vertices and $\tilde O(\sqrt r)$ boundary vertices. The total vertex size of clusters
in an $r$-clustering is $n + \tilde O(n/\sqrt r)$ and the number of clusters is $\tilde O(n/\sqrt r)$.
\begin{lemma}\label{Lem:NestedrClusteringMain}
Let $G$ be a vertex-weighted graph with $m$ edges and $n$ vertices, let $h,\ell\in\mathbb N$ and $0 < \epsilon < 1$ where $h$ and $\epsilon$
are constants and where $\ell = O(\sqrt{n})$ and $\ell = \Omega(n^{\epsilon})$.
For any parameter $r\in(C\log n,\ell]$ for a sufficiently large constant $C$, there is an algorithm with
$O(\sqrt rm^{1 + \epsilon})$ running time that either gives a certificate that $G$ contains $K_h$ as a depth
$O(\ell\log n)$-minor or outputs an $r$-clustering of $G$.
\end{lemma}
\begin{corollary}\label{Cor:SparseMain}
Let $\epsilon > 0$ be a constant. If $\ell = \Omega(n^\epsilon)$, any $n$-vertex graph excluding
$K_h$ as a depth $\ell$-minor has only $O(\mbox{poly}(h)n)$ edges.
\end{corollary}
\begin{lemma}\label{Lem:RMain}
For an $\ell$-clustering $\mathcal C$, $\sum_{C\in\mathcal C}|C||\delta C|, \sum_{C\in\mathcal C}|\delta C|^3 = \tilde O(n\sqrt{\ell})$.
\end{lemma}
For a cluster $C$ and a subset $B$ of $\delta C$, consider the complete undirected graph
$D_B(C)$ with vertex set $B$. Each edge $(u,v)$ in $D_B(C)$ has weight equal to the weight of a shortest path in $C$ between $u$ and
$v$ that does not contain any other vertices of $\delta C$. We call $D_B(C)$ the \emph{dense distance graph of $C$ (w.r.t.\ $B$)}.
We have $d_{D_B(C)}(u,v) = d_C(u,v)$ for all $u,v\in B$. Our algorithm maintains a compact approximate representation of dense
distance graphs using a spanner construction for general graphs; For \emph{stretch}
$\delta\geq 1$, a \emph{$\delta$-spanner} of a graph is a
subgraph that preserves all shortest path distances up to a factor of $\delta$.
\begin{lemma}\label{Lem:SpannerMain}
Let $H$ be an undirected graph with nonnegative edge weights. For any constant $0 < \epsilon \leq 1$,
a $(1/\epsilon)$-spanner of $H$ of size $O(|V(H)|^{1 + 2\epsilon})$ can be constructed in linear time.
\end{lemma}

Now, we give our third and final algorithm. It starts by checking whether $G$ is sparse.
If not, it rejects and halts (Corollary~\ref{Cor:SparseMain}).
Otherwise, the set $V_{\Delta}$ of vertices of degree larger than some value
$\Delta$ are removed. This set has size $O(n/\Delta)$ and will be added as separator vertices in the end.
As we show below, the separator found for the remaining graph will have size $O(n/l) + \tilde{O}(\ell^2\Delta)$ so we pick
$\Delta = \sqrt n/\ell$ to
get a separator for the full graph of size $O(n/\ell) + \tilde O(\sqrt n\ell)$. Note that this is not
an asymptotic increase over the separator size in~\cite{ShallowMinor} when $\ell = \tilde O(n^{1/4})$.
We will give an efficient implementation of algorithm $\texttt{genericalg}$ to find such
a separator in time $O(n^{1+\epsilon}\sqrt\ell)$.

\subsection{Mini clusters}
In the following, we let $G$ denote the graph after the removal of high-degree vertices. Then the degree of $G$
is bounded by $\Delta = \sqrt n/\ell$. We start by
computing an $\ell$-clustering $\mathcal C'$ of $G$. From this, we will obtain a more refined set of
\emph{mini clusters}. Let $C$ be a cluster of $\mathcal C'$. In the following,
we describe how mini clusters associated with $C$ are formed.

For each interior vertex $w$ of $C$, let $C_w$ be the subgraph
of $C$ reachable from $w$ using paths in which no interior vertices belong to $\delta C$.
Let $\mathcal C_1'$ be the set of
subgraphs obtained this way over all $C\in\mathcal C'$. Notice that they are connected and intersect only in
boundary vertices of $C$. We form the subset $\mathcal C_1$ of
$\mathcal C_1'$ consisting of those subgraphs $C_w$ such that there is a pair of distinct vertices $b_1,b_2\in \delta C_w$
where $d_{C_w}(b_1,b_2)$ is smallest among all subgraphs of $\mathcal C_1'$ that have a path between $b_1$ and $b_2$; here we we resolve ties by regarding the path with smaller minimum index of an interior vertex (for some arbitrary fixed assignment of
indices to $V$) as the shortest one.
We say that $(b_1,b_2)$ is \emph{associated} with
$C_w$. Note that any pair of boundary vertices is associated with at most one subgraph of $\mathcal C_1$.

Let $\mathcal C_2$ be the set of edges $e$ such that $e$ has both endpoints in $\delta C$ for some
$C\in\mathcal C'$. We regard each edge of $\mathcal C_2$ as a graph.
Let $\mathcal C = \mathcal C_1\cup \mathcal C_2$ be the set of mini clusters. Note that $\mathcal C$ is
a clustering.
\begin{lemma}\label{Lem:MiniCluster}
With the above definitions, if $G$ excludes $K_h$ as a depth $O(\ell\log n)$-minor,
the total number of boundary vertices (counted with multiplicity)
of mini clusters of a cluster $C\in\mathcal C'$ is $O(|\delta C|\log|\delta C|)$.
\end{lemma}
\begin{proof}
First, note that each subgraph of each cluster is
sparse. If not, it would mean that it does not exclude $K_h$ as a minor~\cite{SparsityHMinor1,SparsityHMinor2}
and hence (since its vertex size is
at most $\ell$) that $G$ does not exclude $K_h$ as a minor of depth $\ell - 1 = O(\ell\log n)$.

Let $\mathcal C_1(C)$ resp.~$\mathcal C_2(C)$ be the set of mini clusters of $\mathcal C_1$ resp.~$\mathcal C_2$
belonging to $C$. Note that the union of mini
clusters of $\mathcal C_2(C)$ is a subgraph of $C$ with vertex set contained in $\delta C$. Since any subgraph of $C$
is sparse, so is this union so
the total number of boundary vertices of mini clusters in $\mathcal C_2(C)$ is $2|\mathcal C_2(C)| = O(|\delta C|)$.

It remains to bound the total number of boundary vertices of mini clusters in
$\mathcal C_1(C)$.
Consider such mini cluster $C'$ and let $w$ be an interior vertex of $C'$. Growing a BFS tree from $w$ in
$C'$, backtracking once a boundary vertex is reached, we get a spanning tree of $C'$ where each boundary vertex is a leaf. The star
with center $w$ and with $\delta C'$ as the set of leaves is a minor of $C'$ as it can be obtained from this spanning tree
using edge contractions. For $i = 1,\ldots,\lceil\log |\delta C|\rceil$, let $\mathcal S_i$ be the set of stars with between $2^{i-1}$ and
$2^i$ leaves over all mini clusters of $\mathcal C_1(C)$ and let $s_i = |\mathcal S_i|$. The lemma will follow if we can show that
$\sum_{1\leq i\leq\lceil\log |\delta C|\rceil}s_i2^i = O(|\delta C|\log|\delta C|)$.

Since mini clusters of $\mathcal C_1(C)$ intersect only in boundary vertices, the union of stars in any set $\mathcal S_i$ is a minor of $C$
and hence excludes $K_h$ as a minor. In particular, this union is sparse so $s_i2^i \leq c(s_i + |\delta C|)$,
for some constant $c$ independent of $i$. Hence, there is a constant $I$ such that
$s_i2^i\leq 2c|\delta C|$ for all $i\geq I$, implying that $\sum_{i\geq I} s_i2^i = O(|\delta C|\log|\delta C|)$.
For $i < I$, consider a star $S\in \mathcal S_i$ and let $C_S$ be the mini cluster that yielded $S$ in the procedure above.
There is a pair of distinct leaves $b_1$ and $b_2$ of $S$ such that $(b_1,b_2)$ is associated with $C_S$. Note that
there is a path from $b_1$ to $b_2$ in $C_S$ such that no interior vertex of this path belongs to any other mini
cluster. Thus, picking such a pair for each
$S\in \mathcal S_i$ and regarding it as an edge, we obtain a minor $H$ of $C$ with vertex set contained in $\delta C$. Since $H$ excludes $K_h$ as a minor (otherwise $C$ would not), we get
$s_i2^i \leq s_i2^I = |E(H)|2^I = O(|\delta C|)$.
\end{proof}
Combining this lemma with Lemma~\ref{Lem:RMain} and the definition of $\ell$-clustering, we get the following corollary.
\begin{corollary}\label{Cor:BoundSzMiniCluster}
With the set $\mathcal C$ of mini clusters defined above, if $G$ excludes $K_h$ as a depth $O(\ell\log n)$-minor,
$\sum_{C\in\mathcal C}|\delta C| = \tilde O(n/\sqrt\ell)$ and
$\sum_{C\in\mathcal C}|C||\delta C| = \tilde O(n\sqrt{\ell})$.
\end{corollary}
Our algorithm will reject and halt if the bounds of Corollary do not hold as then $G$ contains $K_h$ as a depth
$O(\ell\log n)$-minor.

\subsection{Implementing the generic algorithm}
The implementation of $\texttt{genericalg}$ is in many ways similar to our first algorithm so we only highlight the differences here.
Let $S$ be the graph consisting of the union of spanners $S(C)$ over all mini clusters $C\in\mathcal C$,
where the stretch is chosen to be $6/\epsilon$.
We maintain data structures $\mathcal D_i$ for $1\leq i\leq h - 1$ with $d$ of order $\ell\ln n$ in Lemma~\ref{Lem:DistOracle}
and stretch $6/\epsilon$ and each structure is initialized
with graph $S$.
By Lemma~\ref{Lem:DistOracle} and Corollary~\ref{Cor:BoundSzMiniCluster}, this takes
$O(|E(S)|^{1+\epsilon/3}\ell) = \tilde O((n/\sqrt\ell)^{(1 + \epsilon/3)(1+\epsilon/3)}\ell) = O(n^{1+\epsilon}\sqrt\ell)$ time
over all deletions. We will use these structures to identify trees in $G[V\setminus(M\cup A)]$. Note that the stretch of the paths obtained from the structures is
$36/\epsilon^2 = O(1)$.

A problem with this approach is that $G[V\setminus(M\cup A)]$ may contain
parts of mini clusters of $\mathcal C$. Approximate paths in a part of a mini cluster $C$ might not be represented by a subgraph of
$S(C)$ but
each structure $\mathcal D_i$ only supports edge deletions, not insertions. We handle this by requiring the following invariant:
for any mini cluster $C$, if $V(C)\cap V\setminus(M\cup A)\neq\emptyset$ then $V(C)\subseteq V\setminus(M\cup A)$.
If we can maintain this invariant,
at any point, distances in $G[V\setminus(M\cup A)]$ between boundary vertices of $\mathcal C$ can be approximated in a union of spanners
of a subset of the mini clusters in $\mathcal C$. Hence these distances
can be approximated by the structures $\mathcal D_i$ if
we delete all edges of $S(C)$ from them whenever a mini cluster $C$ leaves $G[V\setminus(M\cup A)]$.

In a given iteration, if a tree $T$ is found (line $7$) by a data structure $\mathcal D_i$, it consists of edges from spanners
$S(C)$. By precomputing shortest path trees from each boundary vertex of each mini cluster (which can be done
in $\tilde O(n\sqrt\ell)$ time by Corollary~\ref{Cor:BoundSzMiniCluster}), we can identify the edges of $G$ that are
contained in $T$ in time proportional to their number. Now, to ensure the invariant, we expand $T$ into the
mini clusters it intersects.
More precisely, for each mini cluster $C\in\mathcal C$ for which $V(T)\cap \delta C\neq\emptyset$, we expand $T$ to include all interior vertices of $C$ but no vertices in $\delta C$ that do not already belong to $T$.
This is possible since by definition of mini clusters, there is a spanning tree of $C$ in which every boundary vertex of $C$ is
a leaf.
This tree expansion ensures the invariant but it comes at a cost of a worse separator size guarantee
since trees of the final set $\mathcal M$ may be larger than in our first algorithm. Consider one such tree.
It was obtained by
expanding a tree $T$ of size $O(\ell\log n)$ into the mini clusters of $\mathcal C$ it intersects. By our degree bound,
each vertex of $T$ is contained in at most $\Delta$ mini clusters each having size $O(\ell)$. Hence, each tree in $\mathcal M$ has
size $O(\ell^2\Delta\log n) = O(\ell\sqrt n\log n)$.

We have shown how trees are efficiently found in our implementation of $\texttt{genericalg}$. Lines $10$--$17$ are
implemented exactly as for the first algorithm, where we do BFS in parallel in $G[V']$ from two vertices that are far apart. Note that when these lines are executed, no updates are made to $M\cup A$ (vertices may be removed from $M$
if they are no longer incident to $V'$ but in that case they are added to $A$) so our data structures
$\mathcal D_i$ do not require updates here. The total time for lines $10$--$17$ is $O(m + n\log n)$ as for our first
algorithm. We can now conclude with our final main result.
\begin{theorem}\label{Thm:Algo3}
Given a graph $G$ with $n$ vertices and given $h,\ell\in\mathbb N$ and constants $\epsilon,\epsilon' > 0$ such that with $h = O(1)$ and $\ell = \Omega(n^{\epsilon'})$, there is an algorithm with
$O(n^{1+\epsilon}\sqrt\ell)$ running time which either correctly reports that $G$ contains $K_h$ as a depth $O(\ell\log n)$-minor or
outputs a separator of $G$ of size $O(n/\ell) + \tilde O(\ell\sqrt n)$.
\end{theorem}

\section{Concluding Remarks}\label{sec:ConclRemarks}
We gave three new algorithms to find separators of shallow-minor free graphs. For small-depth minors, a speed-up in running
time of almost a linear factor is achieved compared to previous algorithms while still giving separators of the same size.
A main idea was an application of a dynamic approximate distance oracle of Roditty and Zwick for general graphs. We believe
our speed-ups should give improved algorithms for static graph problems such as shortest paths and maximum matching since
one of the bottlenecks for many separator-based algorithms is finding a good separator. For instance, a faster separator theorem
for minor-free graphs led to several improved static graph algorithms~\cite{CWN}.

Can our bounds be improved further? Nothing suggests that the dependencies on $\ell$ in our time bounds are natural. The
bottleneck is the dynamic distance oracle of Roditty and Zwick. Any improvement of that result would give a speed-up of our algorithms as well.


\newpage
\appendix

\section{Details on the second algorithm}\label{sec:Spanner2}
We now give more details of the algorithm in Section~\ref{sec:Spanner}. Although it is essentially the same as
that in~\cite{CWN}, we include it to make the paper
more self-contained and since we will use a variant of this idea
for the algorithm in Section~\ref{sec:SpannerDistOracle}.

\subsection{Nested $r$-Clustering}
For better readability, we shall repeat the definitions and results from Section~\ref{sec:SpannerDistOracle} here.
For a connected graph $G = (V,E)$, a \emph{cluster} (of $G$) is a connected subgraph of $G$. A \emph{clustering} (of $G$) is a
collection $\mathcal C$ of clusters of $G$ whose edge sets form a partition of $E$.
A \emph{boundary vertex} of a cluster $C\in\mathcal C$ is
a vertex that $C$ shares with other clusters in $\mathcal C$. All other vertices of $C$ are
\emph{interior vertices} of $C$. For a subgraph $C'$ of $C$, we let $\delta C'$ denote the set of boundary vertices of $C$
contained in $C'$ and we refer to them as the boundary vertices of $C'$ (w.r.t.\ $C$).

Let $n$ be the number of vertices of $G$. For a parameter $r > 0$, an \emph{$r$-clustering} (of $G$) is
a clustering with clusters having a total of $\tilde O(n/\sqrt r)$ boundary vertices (counted with multiplicity) and each
containing at most $r$ vertices and $\tilde O(\sqrt r)$ boundary vertices. The total vertex size of clusters
in an $r$-clustering is $n + \tilde O(n/\sqrt r)$ and the number of clusters is $\tilde O(n/\sqrt r)$.

Define a \emph{nested $r$-clustering} of $G$ as follows. Start with an $r$-clustering. Partition each cluster $C$
with a separator of size $O(\sqrt{|V(C)|})$ (we assume such a separator exists) and recurse until clusters
consisting of single edges are obtained. We do it in such a way that both the sizes of clusters and their number of boundary vertices
go down geometrically through the recursion (see Theorem $5$ in~\cite{MultiWeightSep}).

The nested $r$-clustering is the collection $\mathcal C$ of all clusters obtained by this recursive procedure. There are $O(\log r)$
recursion levels and we refer to clusters of $\mathcal C$ at level $i\geq 1$ as \emph{level $i$-clusters} of
$\mathcal C$. Note that level $1$-clusters are those belonging to the original $r$-clustering. We define parent-child
relationships between clusters of $\mathcal C$ in the natural way defined by the recursion.
The following lemma shows the existence of a (nested) $r$-clustering for shallow minor-free graphs and gives an efficient algorithm
to find it. It is done by recursive application of Theorem~\ref{Thm:Algo1}.
\begin{lemma}\label{Lem:NestedrClustering}
Let $G$ be a vertex-weighted graph with $m$ edges and $n$ vertices, let $h,\ell\in\mathbb N$ and $0 < \epsilon < 1$ where $h$ and $\epsilon$
are constants and where $\ell = O(\sqrt{n})$ and $\ell = \Omega(n^{\epsilon})$.
For any parameter $r\in(C\log n,\ell]$ for a sufficiently large constant $C$, there is an algorithm with
$O(\sqrt rm^{1 + \epsilon})$ running time that either gives a certificate that $G$ contains $K_h$ as a depth
$O(\ell\log n)$-minor or outputs a nested $r$-clustering of $G$.
\end{lemma}
\begin{proof}
The proof is essentially the same as in~\cite{CWN} except that we use Theorem~\ref{Thm:Algo1} to find the separators. We therefore
only give a sketch in the following.

We first show how to find an $r$-clustering within the time bound claimed in the lemma.
Let $\epsilon' = \epsilon/3$ and apply Theorem~\ref{Thm:Algo1} with $\ell' = r^{\frac 1 2 - \epsilon'}n^{\epsilon'}$
instead of $\ell$ and $\epsilon'$ instead of $\epsilon$ to obtain a separator $S$ of size
$O(n/\ell') = O(n^{1 - \epsilon'}/r^{\frac 1 2 - \epsilon'})$ in time $O(r^{\frac 1 2 - \epsilon'}n^{\epsilon'}m^{1+\epsilon'})$.
Note that $\ell' < \ell$ so by Lemma~\ref{Lem:ShallowInclusions}, if the algorithm in Theorem~\ref{Thm:Algo1} gives a certificate of a depth
$O(\ell'\log n)$-minor, our algorithm can output this as a certificate of a depth $O(\ell\log n)$-minor.

For each component $C$ of $G\setminus S$, let $B_C$ be the vertices of $S$ incident to $C$ in $G$, and let $H_C = G[V(C)\cup B_C]$.
If an edge belongs to more than one such graph $H_C$ (which can only happen when both its endpoints are boundary vertices), it is removed from
every subgraph except one. Then recurse on every $H_C$ that has more than $r$ vertices. In a recursive call, $n$ above is replaced by the
number of vertices in the current graph.

The result is a clustering $\mathcal C$ of $G$ where each cluster contains at most $r$ vertices and, solving a simple recursive formula,
the total number of boundary vertices
over all clusters (counted with multiplicity) is $\tilde O(n/\sqrt r)$. The time to find this clustering can be expressed by the
following recursion formula, for some constant $d'$:
\[
  T(m,n) \leq d'(r^{\frac 1 2 - \epsilon'}m^{1+2\epsilon'}) + T(m_1,n_1) + T(m_2,n_2),
\]
where $m_1 + m_2 = m$ and $n_1,n_2 < c'n$ for some constant $c' < 1$. It is easy to show that
$T(m,n) = O(r^{\frac 1 2 - \epsilon'}m^{1 + 2\epsilon'}\log n)$.

To get an $r$-clustering, we apply Theorem~\ref{Thm:Algo1} to each cluster $C\in \mathcal C$ that does not have
$\tilde O(\sqrt r)$ boundary vertices, with vertex weights evenly distributed on the boundary vertices of $C$,
and recurse on sub-clusters as above.
Using the same analysis as in~\cite{CWN}, it follows that this gives an $r$-clustering of $G$ within the above time bound.

From the $r$-clustering, we can obtain a nested $r$-clustering by applying Theorem~\ref{Thm:Algo1} to each cluster $C$ with
$\ell' = \sqrt{|V(C)|}$ instead of $\ell$ and recursing on the sub-clusters obtained. Since $\ell'\leq\sqrt r$,
the total time for this is also within that stated in the lemma.
\end{proof}

If we pick $r = \ell$ in Lemma~\ref{Lem:NestedrClustering}, each cluster $C$ in a nested $r$-clustering of $G$ must
exclude $K_h$ as a minor if $G$ excludes $K_h$ as a minor of depth $O(\ell\log n)$ since $|V(C)|\leq \ell$. As shown in~\cite{SparsityHMinor1,SparsityHMinor2}, an $n$-vertex $K_h$-minor-free graph contains only $O(nh\sqrt{\log h})$ edges.
Since level $1$-clusters define a partition of $E$ and their total vertex size
is $n + \tilde O(n/\sqrt r)$, we get the following corollary.
\begin{corollary}\label{Cor:Sparse}
Let $\epsilon > 0$ be a constant. If $\ell = \Omega(n^\epsilon)$, any $n$-vertex graph excluding
$K_h$ as a depth $\ell$-minor has only $O(\mbox{poly}(h)n)$ edges.
\end{corollary}
Note that with this corollary, we can replace $m$ by $n$ in all our time bounds (as $\epsilon$ may be picked
arbitrarily small) if we simply make an initial test to check whether the graph is sufficiently sparse and reject
if not. Note though that this will not enable us to always output a shallow minor as certificate in the case where
a good separator could not be found.

Our algorithm in this section is a slightly modified version of $\texttt{genericalg}$, where in line $7$, we look for a tree in
$G[V']$. This is what is done in~\cite{ShallowMinor}; the algorithm is correct since if a tree cannot be found, a distant pair in
line $10$ must exist.
We will use a nested $r$-clustering together with so called $X$-clusters (defined below) to obtain a compact
sublinear size graph representing $G[V']$. We will then apply Dijkstra on this compact graph
to identify trees efficiently in algorithm $\texttt{genericalg}$.

The following lemma will prove useful later on. It is essentially the same as one from~\cite{CWN} so we omit the proof.
\begin{lemma}\label{Lem:R}
For a nested $\ell$-clustering $\mathcal C$, $\sum_{C\in\mathcal C}|C||\delta C|, \sum_{C\in\mathcal C}|\delta C|^3 = \tilde O(n\sqrt{\ell})$.
\end{lemma}

\subsection{Decomposing subgraphs}
In this subsection, we show how to decompose $G[V']$ into a small number of canonical clusters (which we refer to as $X$-clusters).
Each such cluster will be represented compactly, allowing us to run Dijkstra in sublinear time. We first describe how to use these
clusters to maintain vertex weights of components considered by $\texttt{genericalg}$.

For a given set $X\subseteq V(G)$, define a subset
$\mathcal C_X\subseteq\mathcal C$, obtained by the following procedure. Initialize $\mathcal C_X$ to be the set of
level $1$-clusters of $\mathcal C$. As long as there exists a cluster of $\mathcal C_X$ with
at least one interior vertex belonging to $X$, replace it in $\mathcal C_X$ by its child clusters from $\mathcal C$.

\begin{lemma}\label{Lem:RX}
If $\ell = O(\sqrt{n})$, $\ell = \Omega(n^\epsilon)$, and $|X| = \tilde O(n/\ell)$ then $\mathcal C_X$ consists of clusters sharing only boundary vertices, all vertices of $X$ are boundary vertices in $\mathcal C_X$, and
$\sum_{C\in\mathcal C_X}|\delta C| = \tilde O(n/\sqrt{\ell})$.
\end{lemma}
Again, see~\cite{CWN} for a proof of essentially the same result. The bound on the sum follows by observing that
the worst-case bound
occurs when each of the $\tilde O(n/\sqrt\ell)$ level $1$-clusters contains at most one vertex from $X$. The number of new boundary vertices
introduced when recursively splitting a level $1$-cluster $C$ having a vertex from $X$ is a geometric
sum which is $O(|\delta C|)$. This is $\tilde O(n/\sqrt{\ell})$ over all level $1$-clusters $C$.

Observe that $G[V']$ is the union of connected components of $G\setminus X$, where
$X = M\cup B$, and that $X$ changes during the course of the algorithm. We consider the following
dynamic scenario. Suppose that vertices of $G$ can be in two states, \emph{active} and \emph{passive}. Initially,
all vertices are passive. A vertex can change from passive to active and from active to passive at most once.
At any given point, only $O(\ell\log n + n/\ell)$ vertices are active. If we let $X$ be the set of
active vertices, our modified algorithm satisfies these properties.
We will maintain components of $G\setminus X$ and their vertex weights in this dynamic scenario. We may restrict
our attention to components containing at least one boundary vertex of some cluster; all other components are fully
contained in a single cluster and can easily be handled~\cite{CWN}.

\paragraph{$X$-clusters:}
As $X$ changes, the state of boundary vertices of $C$ may change between active and passive. Refer
to those components of $C\setminus(\delta C\cap X)$ that contain at least one (passive) vertex of $\delta C$
as the \emph{$X$-clusters} of $C$. Any vertex
can change its passive/active state at most twice so the total number of active/passive updates in $\delta C$ is
at most $2|\delta C|$. After each such update, we compute the weights of the new $X$-clusters of $C$. This can be done in
$O(|C|)$ time for each update for a total of $O(|C||\delta C|)$.
By Lemma~\ref{Lem:R}, this is $\tilde O(n\sqrt{\ell})$ over all clusters $C\in\mathcal C$.


\paragraph{Decomposition:}
Now consider the set $\mathcal H$ of components of $G\setminus X$ that contain boundary vertices from $\mathcal C_X$.
We can obtain $\mathcal C_X$ from $\mathcal C$ in time proportional to
$|\mathcal C_X|$ which is $\tilde O(n/\sqrt\ell)$ by Lemma~\ref{Lem:RX}. Each component of $\mathcal H$ is the
union of $X$-clusters in $\mathcal C_X$. We can identify the $X$-clusters forming each component of $\mathcal H$ in time proportional to
$O(|\mathcal C_X|)$ with a standard search procedure. Since we maintain the vertex weight of each $X$-cluster and since by
Lemma~\ref{Lem:RX} the total number of boundary vertices (and hence $X$-clusters) in $\mathcal C_X$ is $\tilde O(n/\sqrt\ell)$,
we can obtain the vertex weight
of each component of $\mathcal H$ within the same time bound (weights of boundary vertices are overcounted but
this can be handled in time proportional to their number).
Combining these results with Lemma~\ref{Lem:NestedrClustering} and the sparsity of minor-free and shallow minor-free
graphs, we get the following.
\begin{lemma}\label{Lem:DecompSubgraphs}
Let $\ell,h\in\mathbb N$ with $\ell = O(\sqrt{n})$, $\ell = \Omega(n^\epsilon)$, and $h = O(1)$ be given, where $\epsilon > 0$ is a constant. There is an
algorithm with $O(n^{1 + \epsilon}\sqrt{\ell})$
preprocessing time which either reports the existence of a $K_h$-minor of depth $O(\ell\log n)$ in $G$ or which, at any point in the dynamic scenario above,
can decompose each component of $G\setminus X$ containing at least one
boundary vertex of $\mathcal C_X$ into $X$-clusters of $\mathcal C_X$ and report the vertex weight of each such component
in a total of $\tilde O(n/\sqrt{\ell})$ time.
\end{lemma}

\subsection{Obtaining a sublinear size graph}
To identify a tree (or find two vertices that are far apart) in $\texttt{genericalg}$, we run Dijkstra in a graph of sublinear
size that contains approximate shortest paths; the same trick is used in~\cite{CWN}.

\paragraph{Dense distance graphs:}
For a cluster $C$ and a subset $B$ of $\delta C$, consider the complete undirected graph
$D_B(C)$ with vertex set $B$. Each edge $(u,v)$ in $D_B(C)$ has weight equal to the weight of a shortest path in $C$ between $u$ and
$v$ that does not contain any other vertices of $\delta C$. We call $D_B(C)$ the \emph{dense distance graph of $C$ (w.r.t.\ $B$)}.
We have $d_{D_B(C)}(u,v) = d_C(u,v)$ for all $u,v\in B$.

We maintain, for each cluster $C\in\mathcal C$, dense distance graph $D_{\delta C\setminus X}(C)$ of the set of
passive boundary vertices of $C$. In a preprocessing step, we compute and store, for each pair of boundary vertices
$u,v\in\delta C$ a shortest path (if any) from $u$ to $v$ in $C$ that does not visit any other boundary vertices (i.e., a shortest
path in $C\setminus(\delta C\setminus\{u,v\})$). For each $u$, we store these paths compactly in
a shortest path tree and we also store the distances from $u$ to each $v$. This takes $O(|C||\delta C|)$
time over all $u$ and $v$ in $\delta C$ which by Lemma~\ref{Lem:R} is $\tilde O(n\sqrt{\ell})$ over all clusters
$C\in\mathcal C$.
From the distances computed, we can obtain dense distance graphs $D_{\delta C}(C)$ over all $C$ within the same time bound.
Dense distance graphs over subsets of $\delta C$ can then be obtained more efficiently when given
$D_{\delta C}(C)$:
\begin{lemma}\label{Lem:PartialDDG}
For a cluster $C$ and $B\subseteq\delta C$, $D_B(C)$ can be obtained from
$D_{\delta C}(C)$ in $O(|B|^2)$ time.
\end{lemma}
\begin{proof}
$D_B(C)$ is the subgraph of $D_{\delta C}(C)$ induced by $B$.
\end{proof}
At the beginning of the dynamic algorithm and whenever a boundary vertex of $C$ changes its state from
passive to active or vice versa, we apply Lemma~\ref{Lem:PartialDDG} to update $D_{\delta C\setminus X}(C)$.
Since there are $O(|\delta C|)$ updates to $\delta C$ in total,
Lemma~\ref{Lem:PartialDDG} shows that the total time for maintaining $D_{\delta C\setminus X}(C)$ is
$O(|\delta C|^3)$. Over all clusters $C\in\mathcal C$, this is $\tilde O(n\sqrt{\ell})$ by Lemma~\ref{Lem:R}.

Let $D_X$ be the union of these dense distance graphs over clusters in $\mathcal C_X$. A shortest path in
$D_X$ between any two boundary vertices has the same weight as a shortest path between them
in $G\setminus X$. Furthermore, $D_X$ has only
$\tilde O(\sqrt{\ell n})$ vertices by Lemma~\ref{Lem:RX}. We sparsify $D_X$ using a multiplicative
spanner construction for general graphs~\cite{SpannerGeneral} to get a graph of size sublinear in $n$.
For $\delta\geq 1$, a $\delta$-spanner of a graph is a subgraph that preserves all shortest path distances
up to a factor of $\delta$.
\begin{lemma}\label{Lem:Spanner}
Let $H$ be an undirected graph with nonnegative edge weights. For any constant $0 < \epsilon \leq 1$,
a $(1/\epsilon)$-spanner of $H$ of size $O(|V(H)|^{1 + 2\epsilon})$ can be constructed in linear time.
\end{lemma}

For each $C\in\mathcal C$, we keep a $(1/\epsilon)$-spanner $S(C)$ of
$D_{\delta C\setminus X}(C)$. Whenever $D_{\delta C\setminus X}(C)$ is updated, we invoke Lemma~\ref{Lem:Spanner}
to update $S(C)$.

Since there are $O(|\delta C|)$ updates to $\delta C$ during the course of the dynamic algorithm,
Lemma~\ref{Lem:Spanner} implies that the total time for maintaining $S(C)$
is $O(|\delta C|^3)$ which over all clusters is $\tilde O(n\sqrt{\ell})$ by
Lemma~\ref{Lem:R}. Let us denote by $S_X$ the graph obtained as the union of $S(C)$ over all $C\in\mathcal C_X$.
Applying Dijkstra in $S_X$ instead of in $G$ gives a speed-up of the algorithm in~\cite{ShallowMinor} using the
same analysis as in~\cite{CWN}: there are $O(n/\ell)$ iterations in total and applying Dijkstra in a given
iteration takes $O(|S_X|) = \tilde O((n/\sqrt\ell)^{1 + 2\epsilon})$ time. This gives Theorem~\ref{Thm:Algo2}.
\end{document}